\documentclass{article}

\usepackage[utf8]{inputenc}
\usepackage{fullpage}
\usepackage{amsmath, amsthm, amssymb,graphicx}
\usepackage{todonotes}
\usepackage{complexity}
\newclass{\TSAT}{3SAT}
\usepackage{stmaryrd}
\usepackage{thmtools, thm-restate}
\usepackage{graphicx}
\usepackage{wrapfig}
\usepackage{xpatch}
\usepackage{authblk}
\usepackage{apptools}
\usepackage{lineno}
\usepackage{subcaption}

\usepackage{url,hyperref}
\hypersetup{
	hidelinks,
	colorlinks=true,
	citecolor=[rgb]{0.121 0.47 0.705},
	linkcolor=[rgb]{0.121 0.47 0.705},
	urlcolor=[rgb]{0.121 0.47 0.705}
}

\graphicspath{{./figs/}}

\newtheorem{theorem}{Theorem}

\newtheorem{observation}{Observation}

\newtheorem{corollary}{Corollary}

\title{Inserting one edge into a simple drawing is hard\thanks{	This work was started during the 6th Austrian-Japanese-Mexican-Spanish Workshop on Discrete Geometry in June 2019 in Austria. We thank all the participants for the good atmosphere as well as discussions on the topic. Also, we thank Jan Kyn\v{c}l for sending us remarks on a preliminary version of this work and an anonymous referee for further helpful comments. 
		A preliminary version of this paper appeared in the Proceedings of the 46th International
		Workshop on Graph-Theoretic Concepts in Computer Science {(WG'20)}}} 

\author[1]{Alan Arroyo\thanks{Email: alanmarcelo.arroyoguevara@ist.ac.at. This project has received funding from the European Union’s Horizon 2020 research and innovation programme under the Marie Sk{\l}odowska-Curie grant agreement No 754411.}}
\author[2]{Fabian Klute\thanks{Email: f.m.klute@uu.nl. Supported by the Netherlands Organisation for Scientific Research (NWO) under project no. 612.001.651 and by the Austrian Science Fund (FWF): J-4510.}}
\author[3]{Irene Parada\thanks{Email: irmde@dtu.dk. Partially supported by the Austrian Science Fund (FWF): W1230, by the collaborative DACH project \emph{Arrangements and Drawings} as FWF project \mbox{I 3340-N35}, and by the Independent Research Fund Denmark grant 2020-2023 (9131-00044B) “Dynamic Network Analysis”.}}
\author[4]{Raimund Seidel\thanks{Email: rseidel@cs.uni-saarland.de.}}
\author[5]{Birgit~Vogtenhuber\thanks{Email: bvogt@ist.tugraz.at. Partially supported by Austrian Science Fund within the collaborative DACH project \emph{Arrangements and Drawings} as FWF project \mbox{I 3340-N35}.}}
\author[6]{Tilo Wiedera\thanks{Email: tilo.wiedera@uos.de. Supported by the German Research Foundation (DFG) grant CH 897/2-2.}}

\affil[1]{IST Austria, Austria}
\affil[2]{Utrecht University, The Netherlands}
\affil[3]{Technical University of Denmark, Denmark}
\affil[4]{Universit\"at des Saarlandes, Germany}
\affil[5]{Graz University of Technology, Austria}
\affil[6]{Osnabr\"uck University, Germany}

\renewcommand{\figurename}{Figure}
\renewcommand{\int}{\ensuremath{\operatorname{int}}}  \newcommand{\snail}{\ensuremath{{\normalfont\bigcirc\hspace{-.89em}\varocircle\hspace{.1em}}}}

\newcommand*\patchAmsMathEnvironmentForLineno[1]{\expandafter\let\csname old#1\expandafter\endcsname\csname #1\endcsname
	\expandafter\let\csname oldend#1\expandafter\endcsname\csname end#1\endcsname
	\renewenvironment{#1}{\linenomath\csname old#1\endcsname}{\csname oldend#1\endcsname\endlinenomath}}\newcommand*\patchBothAmsMathEnvironmentsForLineno[1]{\patchAmsMathEnvironmentForLineno{#1}\patchAmsMathEnvironmentForLineno{#1*}}\AtBeginDocument{\patchBothAmsMathEnvironmentsForLineno{equation}\patchBothAmsMathEnvironmentsForLineno{align}\patchBothAmsMathEnvironmentsForLineno{flalign}\patchBothAmsMathEnvironmentsForLineno{alignat}\patchBothAmsMathEnvironmentsForLineno{gather}\patchBothAmsMathEnvironmentsForLineno{multline}}

\begin{document}

\maketitle

\begin{abstract}
	A {\em simple drawing} $D(G)$ of a graph $G$  is one where each pair of edges share at most one point: either a common endpoint or a proper crossing. An edge $e$ in the complement of $G$ can be {\em inserted} into $D(G)$ if there exists a simple drawing of $G+e$ extending $D(G)$. 
As a result of Levi's Enlargement Lemma, if a drawing is rectilinear (pseudolinear), that is, 
	the edges can be extended into an arrangement of lines (pseudolines), then any edge in the complement of $G$ can be inserted. 
	In contrast,  we show that  it is \NP -complete to decide whether one edge can be inserted into a simple drawing. 
	This remains true even if we assume that the drawing is pseudocircular, that is, 
	the edges can be extended to an arrangement of pseudocircles. On the positive side, we show that, given an arrangement of pseudocircles $\mathcal{A}$ and a pseudosegment $\sigma$,  it can be decided  in polynomial time whether there exists a pseudocircle $\Phi_\sigma$ extending $\sigma$ for which $\mathcal{A}\cup\{\Phi_\sigma\}$ is again an arrangement of pseudocircles.
\end{abstract}

\section{Introduction}
A \emph{simple drawing} of a graph $G$ (also known as \emph{good drawing} or as \emph{simple topological graph} in the literature) 
is a drawing $D(G)$ of $G$ in the plane such that every pair of edges shares at most one point that
is either a proper crossing or a common endpoint. 
In particular, no tangencies between edges are allowed and edges must not contain any vertices in their relative interior. 
It is commonly assumed that no three edges intersect in the same point; the results in this paper are independent of this assumption.
Simple drawings have received a great deal of attention in various areas of graph drawing, 
for example in connection with two long-standing open problems: 
the crossing number of the complete graph~\cite{schaefer2018crossing} and Conway's thrackle conjecture~\cite{problems_book}. 

In this work, we study the problem of inserting an edge into a simple drawing of a graph. 
Given a simple drawing $D(G)$ of a graph $G=(V,E)$ and an edge $e$ of the complement $\overline{G}$ of $ G $ we say that $e$ can be \emph{inserted} into $D(G)$ if there exists a simple drawing of $G' = (V,E\cup \{e\})$ 
that contains $D(G)$ as a subdrawing.

A \emph{pseudoline arrangement} is an arrangement of simple biinfinite arcs, called \emph{pseudolines}, 
such that every pair of pseudolines intersects in a single point that is a proper crossing.
Similarly, an \emph{arrangement of pseudocircles} is an arrangement of simple closed curves, 
called \emph{pseudocircles}, such that every pair of pseudocircles
intersects in either zero or two points, where in the latter case, both intersection points are proper crossings.
A simple drawing~$D(G)$ is called \emph{pseudolinear} if the drawing of every edge can be extended to a pseudoline 
such that the extended drawing forms a pseudoline arrangement.
Recently, Arroyo et al. showed that one can fully characterize these drawings by forbidden  subdrawings and recognize them in polynomial time~\cite{arroyo_socg_2020}.
Likewise, $D(G)$ is called  \emph{pseudocircular} if the drawing of every edge can be extended to a pseudocircle 
such that the extended drawing forms an arrangement of pseudocircles. 

Pseudoline arrangements were introduced by Levi~\cite{levi} in 1926 and have since been extensively studied; see for example~\cite{FelsnerGoodman2017}.
One of the most fundamental results on pseudoline arrangements, nowadays well known as Levi's Enlargement
Lemma, stems from Levi's original paper\footnote{Also known as Levi's Extension Lemma.
		      Several different proofs of Levi's Enlargement Lemma have been published since 
			  then~\cite{pseudo18,Gruenbaum_aas,schaefer2019proof,sweeping_pseudocircles91,SturmfelsZiegler1993}.
		     }.
It states that, for any given pseudoline arrangement $\cal L$ and any two points $p$ and $q$ not on the same pseudoline of $\cal L$, 
it is always possible to insert a pseudoline through $p$ and $q$ into $\cal L$ 
such that the resulting arrangement is again a valid pseudoline arrangement.

From Levi's Enlargement Lemma, it immediately follows that given any pseudolinear drawing $D(G)$ and any set $E^*$ of edges from $\overline{G}$,
it is always possible to insert all edges from~$E^*$ into $D(G)$ such that the resulting drawing is again pseudolinear.
In contrast, if the input drawing $D(G)$ is simple, 
Kyn\v{c}l~\cite{DBLP:journals/dcg/Kyncl13} showed that not every edge of $\overline{G}$
can be added to $D(G)$ such that the result is again a simple drawing,
not even if $G$ is a matching plus two isolated vertices which are the endpoints of the edge to be inserted~\cite{DBLP:journals/comgeo/KynclPRT15}.
The latter implies that an analogous statement to Levi's Enlargement Lemma 
is not true for arrangements of pseudosegments (simple arcs that pairwise intersect at most once).
Moreover, Arroyo, Derka, and Parada~\cite{arroyo_gd_2019} showed that given a simple drawing $D(G)$ 
and a set~$E^*$ of edges from $\overline{G}$, 
it is \NP-complete to decide whether $E^*$ can be inserted into $D(G)$ (such that the resulting drawing is again simple).
However, the cardinality of $E^*$ required for their hardness proof is linear in the size of the constructed graph.
The main open problem posed in~\cite{arroyo_gd_2019} is the complexity of deciding 
whether one single given edge $e$ of~$\overline{G}$ can be inserted into $D(G)$.

In this work, we show that this decision problem is 
\NP-complete, even if~$G$ is a matching plus two isolated vertices which are the endpoints of $ e $.
This implies that, given an arrangement $\cal S$ of pseudosegments and two points $p$ and~$q$ not on the same pseudosegment,
it is \NP-complete to decide whether it is possible to insert a pseudosegment from $p$ to $q$ into $\cal S$ 
such that the resulting arrangement is again a valid arrangement of pseudosegments (Section~\ref{sec:hardness}).
On the positive side, we observe that the decision problem is fixed-parameter tractable (\FPT) in the number of crossings of the original drawing~$G$ (Section~\ref{sec:fpt}). 
This algorithm cannot be directly adapted to obtain an \FPT-algorithm only with respect to the number of newly created crossings. 
Very recently, an overlapping set of authors showed an \FPT-algorithm for this problem that is tight under the Exponential Time Hypothesis~\cite{ganian2020crossingoptimal}. 
Using a different approach that requires invoking Courcelle's theorem~\cite{Courcelle90}, 
the authors present an \FPT-algorithm for inserting a bounded number of edges with a bounded number of new crossings into a simple drawing~$G$. 

Snoeyink and Hershberger~\cite{sweeping_pseudocircles91} showed the following analog to Levi's Enlargement Lemma for arrangements of pseudocircles: 
For any arrangement~${\cal A}$ of pseudocircles and any three points $p$, $q$, and $r$, 
not all of them on one pseudocircle of ${\cal A}$, there exists a pseudocircle $\Phi$ through $p$, $q$, and $r$ such that ${\cal A} \cup \{\Phi\}$ is again an arrangement of pseudocircles.
Refining our hardness proof, we show that the edge-insertion decision problem remains \NP-complete when
$D(G)$ is a pseudocircular drawing, 
regardless of whether the resulting drawing is required to be again pseudocircular or allowed to be any simple drawing.
This holds even if we are in addition given an arrangement of pseudocircles extending~$D(G)$.
On the positive side, we show that, given an arrangement~${\cal A}$ of pseudocircles and a pseudosegment $\sigma$, 
it can be decided in polynomial time whether there exists an extension~$\Phi_\sigma$ of $\sigma$ 
to a simple closed curve such that ${\cal A} \cup \{\Phi_\sigma\}$ is again an arrangement of pseudocircles (Section~\ref{sec:arr_pseudocicles}).

\paragraph{More related work.}
One of the implications of the results presented in this paper 
concerns so-called saturated drawings~\cite{DBLP:journals/comgeo/KynclPRT15}. 
A simple drawing $D(G)$ of a graph $G$ is called \emph{saturated}
if no edge $e$ from $\overline{G}$ can be inserted into $D(G)$.
Kyn{\v{c}}l et al. showed that there are saturated simple drawings whose number of edges is only linear in the number of vertices~\cite{DBLP:journals/comgeo/KynclPRT15}.
The currently best upper bound on the minimum number of edges in saturated simple drawings is $7n$ and
has been shown by Hajnal et al.~\cite{Hajnal2015saturated}.
A natural question is to determine the complexity of deciding whether a simple drawing is saturated.
Our hardness result implies that the straight-forward idea of testing whether $D(G)$ is saturated
by checking for every edge in $\overline{G}$ whether it can be inserted into $D(G)$
is not feasible unless $\P = \NP$.

The problem of inserting an edge (or multiple edges or a star) into a planar graph has been extensively studied in the contexts of determining the crossing number of the resulting graph~\cite{AddingOneEdge_Cabello_2013,crossingnumbercubic_Riskin_1996} and of finding a drawing of the resulting graph in which the original planar graph is drawn crossing-free and the drawing of the resulting graph has as few crossings as possible~\cite{DBLP:conf/soda/ChimaniGMW09,DBLP:conf/compgeom/ChimaniH16,InsertingEdgePlanar_Gutwenger_2005,DBLP:conf/gd/RadermacherR18}. 
 In relation to our work, a main difference is that we consider inserting edges into some given non-plane drawing of a graph.

Furthermore, the question considered in this paper is strongly related to work on extending partial representations 
of graphs.
Here, we are usually given a representation of a part of the graph $G$ 
and are asked to extend it into a full representation of $G$
such that the partial representation is a sub-representation of the full one.
Recent years have seen a plethora of results in this topic.
For plane drawings Angelini et al.~\cite{DBLP:journals/talg/AngeliniBFJKPR15} showed that the problem can be solved in linear time, 
while Patrignani already proved earlier that the problem is \NP-complete for plane straight-line drawings~\cite{ext_straight_06}. 
For level and upward planar graph drawings the problem was shown to be \NP-complete~\cite{PartialConstrainedLevel_Brueckner_2017,lozzo_compgeo_2020}. 
However, under certain restrictions on the graph and the drawing, the extension problems become tractable~\cite{PartialConstrainedLevel_Brueckner_2017,cegl-dgpwpofpa-12,lozzo_compgeo_2020,ExtendingConvexPartial_Mchedlidze_2015}. 
Very recently, also orthogonal drawings have been considered~\cite{angeliniExtendingPartialOrthogonal2020}.
Extension of other graph respresentations have been studied for several graph classes defined by intersection or visibility of geometric objects~\cite{ContactRepresentationsPlanar_Chaplick_2014,Extendingpartialrepresentations_Chaplick_2019,Chaplick2018,ExtendingPartialRepresentations_Klavik_2012,ExtendingPartialRepresentations_Klavik_2017a,Extendingpartialrepresentations_Klavik_2015,ExtendingPartialRepresentations_Klavik_2017}.
Very recently, the extension problem was also considered for 1-plane drawings
through the lens of parameterized complexity~\cite{eiben_mfcs_2020,eiben_icalp_2020}.

A similar extension problem was studied when the graph class considered are trees.
	Here, we are also given a point-set $P$ and
	ask if the given drawing can be extended using only points in $P$ for vertex positions.
	Di Giacomo et al.~\cite{giacomo_trees_09} showed that this problem is polynomial time solvable if
	bends are allowed.
	Similarly to the case of planar graphs,
	Bagheri and Razzazi~\cite{Planarstraightline_Bagheri_2010} showed that the problem is \NP-complete 
	when we require the extended drawings to be straight-line.

\paragraph{Outline.}
The remainder of our paper is organized as follows.
In Section~\ref{sec:hardness} we prove 
that, given a simple drawing $ D(G) $ of a graph $ G$,
it is \NP-complete to decide whether a given edge $ e $ of $ \overline{G} $ can be inserted into $ D(G) $. 
Furthermore, we discuss under which conditions 
the statement holds. 
Most notably, in Section~\ref{sec:extension}, we show that the problem remains \NP-hard even if the input drawing is pseudocircular.
In contrast, we show in Section~\ref{sec:arr_pseudocicles} that 
for a given arrangement ${\cal A}$  of pseudocircles and a pseudosegment $\sigma$, 
we can decide in polynomial time whether $ \sigma $ can be extended to simple closed curve $\Phi_\sigma$
such that ${\cal A} \cup \{\Phi_\sigma\}$ is again an arrangement of pseudocircles.
Finally, in Section~\ref{sec:fpt}, we observe that 
the problem of deciding whether a given edge $ e $ of $ \overline{G} $ 
can be inserted into a simple drawing $ D(G) $ of a graph $ G $
is \FPT\ 
in the number of crossings of $ D(G) $.

\section{Inserting one edge into a simple drawing is hard}
\label{sec:hardness}

In this section we prove the following theorem containing our main result:
\begin{theorem}\label{thm:oneedgehard}
	Given a simple drawing $ D(G) $ of a graph $ G = (V,E) $ and an edge $ uv $ of $\overline{G}$, it is \NP-complete to decide whether $ uv $ can be inserted into $ D(G) $,
	even if $ V \setminus \{u,v\} $ induces a matching in $ G $ and $ u $ and $ v $ are isolated vertices.
\end{theorem}

It is straightforward to verify that the problem is in \NP\ (see Arroyo et al.~\cite{arroyo_gd_2019} for a combinatorial description of our problem using the dual of the planarization of the drawing). We show \NP-hardness via a reduction from \TSAT. 
Let $ \phi(x_1,\ldots x_{n}) $ be a \TSAT-formula with \emph{variables} $ x_1,\ldots, x_{n} $ and set of \emph{clauses} $\mathcal C = \{C_1,\ldots, C_{m}\}$.
An occurrence of a variable~$ x_i $ in a clause $ C_j \in \mathcal C $ is called a \emph{literal}. 
For convenience, we assume that in $ \phi(x_1, \ldots, x_{n}) $, each clause has three (not necessarily different) literals.
In a preprocessing step, 
we eliminate clauses with only positive or only negative literals 
via the transformation from Lemma~\ref{clm:transform}.

\begin{restatable}{lemma}{claimtransform}\label{clm:transform}
The following transformation of a clause with only positive or only negative literals, respectively, preserves the satisfiability of the clause
	($y$ is a new variable and $\mathtt{false}$ is the constant value false): 
\begin{align*}
		x_i \! \lor \! x_j \! \lor \! x_k \Rightarrow \!
		\begin{cases}
			x_k \! \lor \!  y \lor \mathtt{false} & \! \text{(i)}\\
			x_i \! \lor \!   x_j \! \lor \! \neg y & \! \text{(ii)}
		\end{cases}
		& & \hspace{-2.75mm}
		\neg x_i \! \lor \! \neg x_j \! \lor \! \neg x_k \Rightarrow \!
		\begin{cases}
			\neg x_i  \! \lor  \! \neg x_j  \! \lor  \! y & \! \text{(iii)}\\
			\neg x_k  \! \lor  \! \neg y  \! \lor  \! \mathtt{false} & \! \text{(iv)}
		\end{cases}
	\end{align*}
\end{restatable}
\begin{proof}
	We prove the statement for the case in which the original clause has three positive literals; the other case is analogous.
	Assume that $ x_i $ or $ x_j $ satisfies the original clause. 
	Then it also satisfies Clause (ii) and $ y $ can be set to $ \texttt{true} $ to  satisfy Clause (i).
	If $ x_k $ satisfies the original clause, then it also satisfies Clause (i) and $ y $ can be set to $ \texttt{false} $ to satisfy Clause~(ii).
	If none of $ x_i $, $ x_j $, and $ x_k $ satisfy the original clause, 
	then to satisfy Clause (ii) we have to set $ y $ to $ \texttt{false} $,
	which implies that Clause (i) is not satisfied.
\end{proof}

After the preprocessing, 
we have a \emph{transformed} \TSAT-formula where each clause is of one of the following four types:
Type (i) two positive literals and one constant $ \texttt{false} $; 
Type (ii) one negative and two positive literals; 
Type (iii) one positive and two negative literals,
and finally, Type (iv) two negative literals and one constant~$ \texttt{false} $.

Given a transformed \TSAT-formula $ \phi = \phi(x_1, \ldots ,x_n) $ with set of clauses $ \mathcal C = \{C_1, \ldots , C_m\} $,
satisfiability of $ \phi $ will correspond to being able to insert a given edge $ uv $ 
into a simple drawing~$ D $  of a matching constructed from the formula $\phi$.
The main idea of the reduction is that the variable and clause gadgets in $D$ act as ``barriers'' inside a simple closed region $ R $ of $ D $,
in which we need to insert a simple arc $\gamma$ from one side to the other to connect $ u $ and~$ v $. 
Crossing a barrier in some way imposes constraints on how or whether we can cross other barriers afterwards.

To simplify the description, we first focus our attention to the inside of the simple closed region~$ R $.
We assume that $\gamma$ cannot cross the boundary of $R$.
In the following we use two lines, named $\lambda$ and $\mu$,
to bound the regions in which a variable and clause gadget will be placed. 
Particularly, these lines will be identified with opposite segments on $R$'s boundary.

\paragraph{Variable gadget.}
A variable gadget $ W $ is bounded from the left by a vertical line $ \lambda $ and from right by a vertical line $ \mu $.
Additionally, it contains a horizontal segment $\kappa$ between $\lambda$ and $\mu$,
a set~$P$ of pairwise non-crossing arcs (parts of later-defined edges), each with one endpoint on $\kappa$ and the other endpoint on~$\mu$,
and a set $N$ of pairwise non-crossing arcs, each with one endpoint on $\kappa$ and the other endpoint on $\lambda$.
On $\kappa$, all the endpoints of arcs in $P$ lie above
all the endpoints of arcs in $N$, implying that 
every arc in $ P $ crosses every arc in~$ N $. 
Finally, we choose two points $ u $ and $ v $ such that $ u $ is below all arcs in~$ W $ and $ v $ is above them; 
see \figurename~\ref{fig:variable} for an illustration.
The arcs in $ P $ and $ N $ correspond to positive and negative appearances of the variable, respectively.

\begin{figure}
	\begin{minipage}[t]{.48\textwidth}
		\centering
		\includegraphics[page=1]{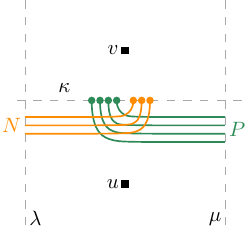}
		\caption{Variable gadget. Orange arcs belong to $ N $, green ones to $ P $.}
		\label{fig:variable}
	\end{minipage}
	\hfill
	\begin{minipage}[t]{.48\textwidth}
		\centering
		\includegraphics[page=1]{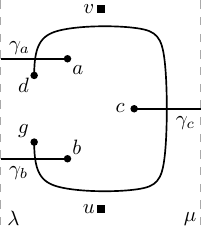}
		\caption{Clause gadget.}
		\label{fig:clause}
	\end{minipage}
\end{figure}

\begin{restatable}{lemma}{lemvariable}\label{lem:variable}
	Let $ W $ be a variable gadget. 
	Any arc between the vertical lines $ \lambda $ and $ \mu $ that connects $ u $ and $ v $
	crosses all arcs in $ P $ or all arcs in $ N $.
\end{restatable}
\begin{proof}
	Assume that there is an arc connecting $ u $ and $ v $ neither crossing all the arcs in $ P $ nor all the arcs in $ N $.
	Hence, there are two arcs $ p \in P $ and $ n \in N $ such that this arc neither crosses $ p $ nor $ n $.
	By the construction of the gadget, $ p $ and $ n $ cross.
	Thus, their union together with $ \lambda $ and $ \mu $ separates $ u $ from $ v $.
	It follows that the arc has to cross $ p $ or $ n $.
\end{proof}

\paragraph{Clause gadget.} 
Similar to a variable gadget, a clause gadget $ K $ is bounded from the left and right by two vertical lines $\lambda$ and $\mu$, respectively.
Additionally, it contains three horizontal arcs (parts of later-defined edges) $ \gamma_a $, $ \gamma_b $, and $ \gamma_c $,
where the former two have one endpoint on $\lambda$ and the latter has one endpoint on $\mu$.
On~$\lambda$, the endpoint of $ \gamma_a $ lies to the right of the one of $ \gamma_b $.
The other endpoints of $ \gamma_a $, $ \gamma_b $, and $ \gamma_c $ are called $a$, $b$, and $c$, respectively.
None of these three arcs cross.
Moreover, $K$ contains two points $d$ and $g$ and an edge  $ dg $ that crosses $ \gamma_a $, $ \gamma_c $, and $ \gamma_b $ in that order when traversed from $ d $ to $ g $.
Notice that we do not require any specific rotation of the crossings of $ dg $ with $ \gamma_a $ and $ \gamma_b $
(where the rotation is the clockwise order of the endpoints of the crossing arcs). 
However, to simplify the description, we assume that the rotations of the crossings are as in \figurename~\ref{fig:clause}.
The rotation of the crossing of $ dg $ with $ \gamma_c $ is forced by the order of the crossings along $dg$.
Finally, we again choose two points $ u $ and $ v $ such that $ u $ is below all arcs in $ K $ and $ v $ is above them; 
see \figurename~\ref{fig:clause} for an illustration.

\begin{restatable}{lemma}{lemclause}\label{lem:clause}
	Let $ K $ be a clause gadget.
	Any arc $ uv $ between the vertical lines $ \lambda $ and $ \mu $ that connects $ u $ and $ v $ 
	crosses either $dg$ twice or 
	at least one of the arcs $ \gamma_a $, $ \gamma_b $, and $ \gamma_c $.
\end{restatable}
\begin{proof}
	Let $ \times $ be the crossing point of $ \gamma_c $ and $ dg $.
	This point splits the arc $ dg $ into two arcs $ d\times $ and $ g\times $.
	Assume that the arc $ uv $ does not cross the arcs $ \gamma_a $, $ \gamma_b $, and $ \gamma_c $.
	The union of $ \gamma_a $ and $ \gamma_c $
	together with $ d\times $
	and the lines $ \lambda $ and $ \mu $
	separates $ u $ from $ v $.
	Since the arcs $ \gamma_a $ and $ \gamma_c $ are not crossed by $ uv $,
$ uv $ must cross $ d\times $ in a point $ \times' $. Analogously, the union of $ \gamma_b $, $ \gamma_c $, 
	together with $ g\times $
	and the lines $ \lambda $ and $ \mu $
	separates $ u $ from $ v $.
	Thus, $ uv $ has to cross $ g\times $ in a point  $ \times'' \neq \times' $ to avoid tangencies. This implies that $ uv $ crosses $ dg $ twice, a contradiction.
\end{proof}

\paragraph{The reduction.} 
Let $ \phi(x_1, \ldots ,x_n) $ be a transformed \TSAT-formula with clause set $ \mathcal C = \{C_1, \ldots , C_m\} $ (each clause being of one of the four types identified above).
To build our reduction we need one more gadget.
First, we introduce the following simple drawing introduced by Kyn\v{c}l et al.~\cite[Figure~11]{DBLP:journals/comgeo/KynclPRT15}
and depicted in \figurename~\ref{fig:kyncl}.
Here, we denote this drawing by $ \snail $.
Following the notation by Kyn\v{c}l et al., we denote its six arcs by $a_1$, $a_2$, $a_3$, $b_1$, $b_2$, and~$b_3$; and 
its eight cells by $ X $, $ A_1 $, $ A_2 $, $ A_3 $, $ B_1 $, $ B_2 $, $ B_3 $, and $ Y $; see \figurename~\ref{fig:kyncl} for an illustration. 
The core property $ \mathcal P $ of $ \snail $ is that it is not possible to insert an edge between a point in cell~$ X $ and another point in cell $ Y $ 
such that the result is a simple drawing~\cite[Lemma 15]{DBLP:journals/comgeo/KynclPRT15}.

For our reduction, we first choose two arbitrary points $ u $ and $ v $ in the cells~$ X $ and $ B_2 $ and insert them as vertices into $ \snail $. 
Let  $ \snail' $ be the obtained drawing. 
Further, let $ b_2^* $ be the part of the arc $ b_2 $ between the crossing point of $ b_2 $ and $ a_2 $ and
the crossing point of $ b_2 $ and~$ b_3 $, see again \figurename~\ref{fig:kyncl}.

\begin{figure}
	\begin{minipage}[t]{.41\textwidth}	
		\centering
		\includegraphics[page=1]{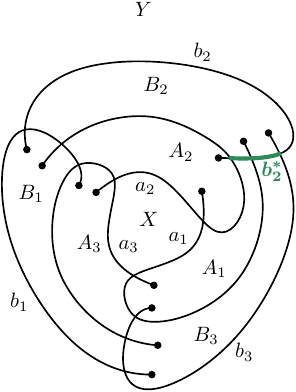}
		\caption{The simple drawing \snail\ presented in~\cite{DBLP:journals/comgeo/KynclPRT15}.
			It is not possible to insert an edge between a point in $X$ and one in $Y$.}
		\label{fig:kyncl}
	\end{minipage}
	\hfill
	\begin{minipage}[t]{.55\textwidth}	
		\centering
		\includegraphics[page=2]{kyncl}
		\caption{A schematic overview of the edges in $F$ (red and orange) and how they are combined with~\snail.}
		\label{fig:overview}
	\end{minipage}
\end{figure}

\begin{restatable}{lemma}{lemboundary}\label{lem:boundary}
	The edge $ uv $ cannot be inserted into $ \snail' $ without crossing $ b_2^* $.
\end{restatable}
\begin{proof}
Assume for contradiction that $uv$ can be inserted not crossing $ b_2^* $ and let $\gamma_{uv}$ be such an arc. 
	Refer to \figurename~\ref{fig:kyncl}. 
	If~$\gamma_{uv}$  does not cross $b_2$, then we would be able to prolong it 
	and cross $b_2$ to reach $Y$, a contradiction of property $\mathcal{P}$.
	Thus, $\gamma_{uv}$ crosses $b_2$. 
	Further, we may assume without loss of generality that $\gamma_{uv}$ does not cross $b_2$ inside $A_2$ or $B_1$, as otherwise it would be possible to modify
	$\gamma_{uv}$ to not cross $b_2$. 
	Thus, $\gamma_{uv}$ intersects $b_2$ on the boundary of $B_2$.
	Since $\gamma_{uv}$ cannot intersect $Y$, this crossing must be on $ b_2^* $. 
\end{proof}

The final piece we need for our reduction is a set $ F $ of $ m^{I} + m^{IV} + 4 $ arcs that we insert into $ \snail' $,
where $ m^{I} $ is the number of clauses of Type (i) and
$ m^{IV} $ the number of clauses of Type (iv).
For an arc $ f \in F $ we will place one of its endpoints on a vertical line $ \kappa_F $ inside~$ A_2 $
and the other one inside $ B_2 $;
see \figurename~\ref{fig:overview} for an illustration.
The only crossings of $ f $ with $ \snail' $ are with the arcs $ a_2 $, $ a_1 $, $ b_3 $, and $ b_2 $, in that order,
when traversing $ f $ from its endpoint on $ \kappa_F $ to its endpoint in $ B_2 $.
Furthermore, when $ f $ is traversed in that direction, it crosses from $ A_2 $ to $ A_1 $, from $ A_1 $ to $ B_3 $, from $ B_3 $ to $ Y $, and from $ Y $ to $ B_2 $.

Consider the $ m^{I} + m^{IV} + 4 $ endpoints on $ \kappa_F $ sorted from top to bottom. 
We denote by~$f_j$ the arc in $F$ incident with the $j$-th such endpoint. 
When traversing $ b_2 $ from its endpoint in $ A_2 $ to its endpoint in $ B_1 $,
the crossings of arcs in $ F $ with $ b_2 $ appear in the same order as their endpoints on $ \kappa_F $. 
More precisely, the crossings of $ b_2 $, when $b_2$ is traversed in that direction, are with $a_2$, $a_1$, $b_3$,  $f_1$,  $f_2$, \ldots, $ f_{|F|}$, and $b_1$, in that order. 

The arcs $f_{m^{I} + 1}$, $f_{m^{I} + 2}$, $f_{m^{I} + 3}$, and $f_{m^{I} + 4}$ 
will behave differently than the other arcs in~$ F $. 
In the following, we denote these four arcs by $ r_2 $, $ r_1 $, $ \ell_1 $, and $ \ell_2 $, respectively.
There are only two crossings between arcs in $ F $,
namely, between $ r_1 $ and $ r_2 $, and between $ \ell_1 $ and~$ \ell_2 $, 
and both these crossings are inside $ B_2 $. 
These four crossing arcs divide $B_2$ into three regions. 
Let $R$ denote the region with $b_2^*$ on its boundary; 
let $R_r$ denote the (other) region incident with the crossing between  $ r_1 $ and $ r_2 $; 
and let $R_\ell$ denote the (other) region incident with the crossing between $\ell_1$ and $\ell_2$.
Arcs $ r_1 $, $ r_2 $, $ \ell_1 $, and $ \ell_2 $ must be drawn such that the vertex $v$ lies in $R$; 
see the red arcs in \figurename~\ref{fig:overview} for an illustration.
The precise endpoints of the edges in $ F \setminus \{r_1,r_2,\ell_1,\ell_2\} $ will be fixed when we insert the clause gadgets.

\begin{restatable}{lemma}{lemblocker}\label{lem:blockers}
	The edge $ uv $ cannot be inserted into $ \snail' $ without crossing every arc in $ F $ in $ A_1$ or $B_3 $ (in the interior or common boundary of these cells).
\end{restatable}
\begin{proof}
	Assume for contradiction that there is an arc $ f \in F $ such that $ uv $ does not cross $ f $.
	From Lemma~\ref{lem:boundary} we know that $ uv $ has to cross $ b_2^* $.
	Consider the region bounded by $ b_2^* $, $ b_3 $, $ f $, and $ a_2 $. 
Observe that, since $ b_2^* $ is fully contained on the boundary of this region, 
	$ uv $ has to cross at least one of the three other arcs as well.
	By assumption, $ uv $ does not cross $ f $.
	Crossing $ b_3 $ is impossible by property $\mathcal{P}$, as the part contained on this region's boundary separates $ B_3 $ from~$ Y $.
	Finally, crossing the arc which is part of $ a_2 $ is not possible,
	since this would imply the existence of a point $ v' $ in $ A_2 $ such that $ uv $ passes through $ v' $ without having crossed $ a_2 $.
	Hence, we could prolong the arc $ uv' $ that is part of~$ uv $  
	by crossing $ a_2 $ such that it reaches $ B_2 $ without crossing $ b_2^* $, a contradiction to Lemma~\ref{lem:boundary}.
Thus, the statement follows. 
\end{proof}

\begin{figure}
	\centering
	\includegraphics[page=1]{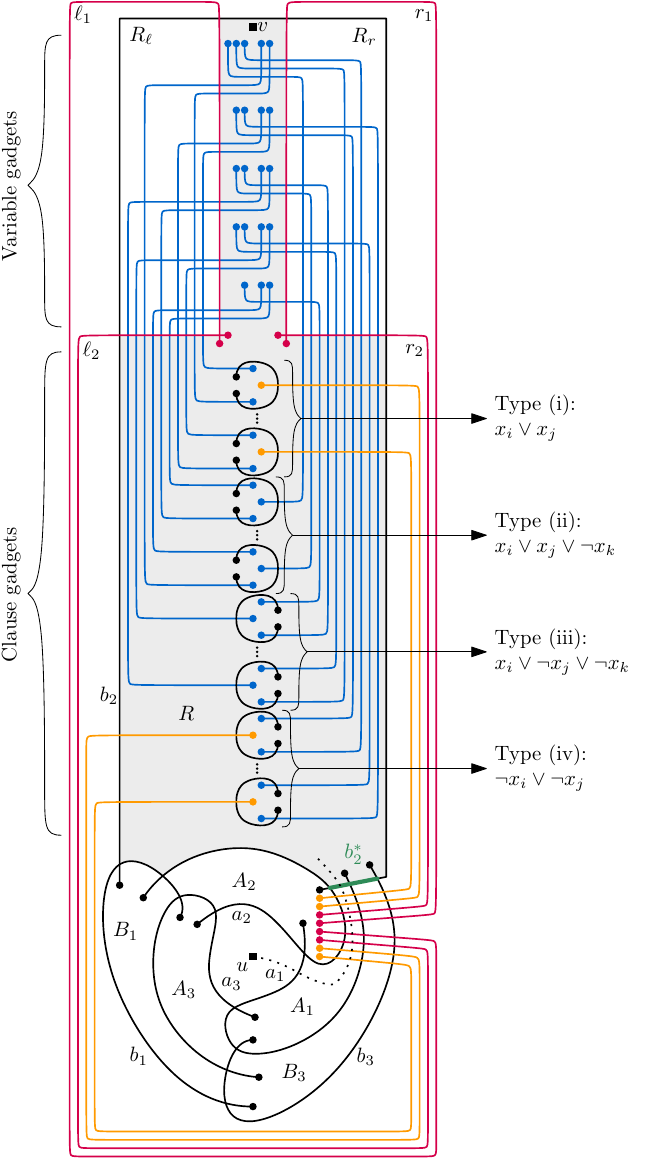}
	\caption{Illustration of the reduction. }
	\label{fig:reduction}
\end{figure}

It remains to insert inside $ R $ the clause and variable gadgets and precisely define the endpoints of arcs in $ F \setminus \{\ell_1,\ell_2,r_1,r_2\}$.
For simplicity, we first insert the variable gadgets and then the clause gadgets. 
The idea is that each clause and variable gadget is inserted in~$R$ separating $ b_2^* $ from $ v $.
This is done by identifying the endpoints that were lying on $ \lambda $ or~$ \mu $ with points on $ \ell_1 $, $ \ell_2 $, $ r_1 $, $ r_2 $, or $b_2$.
As a result, Lemmas~\ref{lem:variable} and~\ref{lem:clause} can be applied to the arc 
that we insert connecting $ u $ and $ v $ in the final drawing,
since it has to cross~$ b_2^* $ by Lemma~\ref{lem:boundary}.

We now insert the variable gadgets into $R$. 
Let $ W^{(i)} $ be the variable gadget corresponding to variable $ x_i $.
For a gadget $ W^{(i)} $, 
the arcs in $ N $  are drawn such that the endpoints on $ \lambda $ 
lie on the part of $ \ell_1 $ that bounds $R$. 
The arcs in $ P $ are drawn similarly, but with the endpoints on $ \mu $ 
lying on the part of $ r_1 $ that bounds~$R$. 
Moreover, we identify vertex $v$ in the gadget with vertex $v$ in~$\snail '$. 
Gadgets corresponding to different variables are inserted without crossing each other. 
We now specify how they are inserted relative to each other. 
As we traverse $\ell_1$ from its endpoint on $\kappa_F$ to its endpoint in $R$,
we encounter the endpoints of arcs in $ W^{(i)} $ before 
the endpoints of arcs in $ W^{(i+1)} $.
Analogously, as we traverse $r_1$ from its endpoint on~$\kappa_F$ to its endpoint in $R$,
we encounter the endpoints of arcs in $ W^{(i)} $ before 
the endpoints of arcs in $ W^{(i+1)} $. 
See \figurename~\ref{fig:reduction} for an illustration.

In a similar way we insert the clause gadgets.
Let $ K^{(j)} $ be the clause gadget corresponding to clause $ C_j $. 
If $C_j$ is of Type~(i), $ K^{(j)} $ is inserted such that 
the endpoints on $\lambda$ lie on the part of $ \ell_2$ that bounds $R$.  
If $C_j$ is the $j'$-th clause of Type~(i), 
we identify $c$ with the endpoint of the arc $f_{j'}$. 
Similarly, if $C_j$ is of Type~(iv), 
$ K^{(j)} $ is inserted such that 
the endpoints on $\lambda$ lie on the part of~$r_2$ that bounds $R$.  
If $C_j$ is the $j'$-th clause of Type~(iv), 
we identify $c$ with the endpoint of the arc $f_{m^{I} + 4 + j'}$. 
If $C_j$ is of Type~(ii), $ K^{(j)} $ is inserted such that 
the endpoints on $\lambda$ lie on the part of $ \ell_2$ that bounds $R$ and 
the endpoint on $\mu$ lies on the part of~$ r_2$ that bounds $R$.  
Similarly, if $C_j$ is of Type~(iii), $ K^{(j)} $ is inserted such that 
the endpoint on $\mu$ lies on the part of $ \ell_2$ that bounds $R$ and 
the endpoints on $\lambda$ lie on the part of~$ r_2$ that bounds $R$.  
The crossings in $R$ of arcs from different clause gadgets 
are of arcs with an endpoint in $r_2$ with arcs in $\{f_j: 1\leq j \leq m^{I}\}$.

We now specify how different clause gadgets  are inserted relative to each other. 
As we traverse $\ell_2$ from its endpoint on $\kappa_F$ to its endpoint in $R$,
we first encounter the endpoints of arcs corresponding to Type~(iii) clauses,
followed by the ones corresponding to Type~(ii) clauses, 
and finally the ones corresponding to Type~(i) clauses.
Analogously, as we traverse $r_2$ from its endpoint on $\kappa_F$ to its endpoint in $R$,
we first encounter the endpoints of arcs corresponding to Type~(iv) clauses,
followed by the ones corresponding to Type~(iii) clauses, 
and finally the ones corresponding to Type~(ii) clauses.
Moreover, as we traverse $\ell_2$ and $r_2$ in the specified directions, 
the endpoints of arcs corresponding to the $j'$-th clause of a certain type 
are encountered before the endpoints of arcs corresponding to the $(j'-1)$-st clause of this type. 
An illustration can be found in \figurename~\ref{fig:reduction}.

Finally, we connect arcs from variable and clause gadgets inside the regions $ R_\ell $ and $ R_r $. 
This is done such that if a literal in a clause is $x_k$ then the corresponding arc in the clause gadget, 
that has an endpoint on $\ell_2$, is connected with an arc in~$N$ of the gadget $W^{(k)}$, 
that has an endpoint on $\ell_1$. 
Thus, these connections can lie in $R_\ell$. 
Analogously, if a literal in a clause is $\neg x_k$ then the corresponding arc in the clause gadget, 
that has an endpoint on $r_2$, is connected with an arc in~$P$ of the gadget $W^{(k)}$, 
that has an endpoint on $r_1$. 
Thus, these connections can lie in $R_r$. 
Since, without loss of generality, we can assume that $R_\ell$ and $R_r$ are convex regions 
and the endpoints we want to connect are pairwise distinct points on the boundaries of those regions, the connections can be drawn as straight-line segments. 
(For visual clarity in \figurename~\ref{fig:reduction} and to argue pseudocircularity in Section~\ref{sec:extension}, we draw these connections with one bend per arc.)
Therefore, there is at most one crossing between each pair of connecting arcs.

Each connecting arc is concatenated with the arcs in a variable and in a clause gadget that it joins. 
These concatenated arcs are edges in our drawing that have one endpoint in a variable gadget and the other one in a clause gadget. 
By construction, each such edge corresponds to a literal in the formula $\phi$ 
and each pair of them crosses at most once.  
Similarly, the arcs in $F\setminus \{\ell_1, \ell_2, r_1, r_2\}$ have one endpoint in 
a clause gadget and also define edges in our final drawing 
that we denote by the same names as the corresponding arcs.

We now have all the pieces that constitute our final drawing. 
It consists of 
(i)  the simple drawing~$\snail'$; 
(ii) the edges $f_i\in F$ drawn as the described arcs (with their endpoints as vertices); 
(iii) the edges corresponding to literals (with their endpoints as vertices); and
(iv) the edges $dg$ in each clause gadget (with $d$ and~$g$ as vertices). 
Observe that the constructed drawing is a simple drawing, 
as it is the drawing of a matching (plus the vertices $u$ and $v$) and, 
by construction, any two edges cross at most once.

It remains to show that the presented construction is a valid reduction.

\begin{restatable}{lemma}{lemcorrectness}
	\label{lem:correctness}
	The above construction is a polynomial time reduction from \TSAT\ 
	to the problem of deciding whether an edge can be inserted into a simple drawing.
\end{restatable}
\begin{proof}
	Given a \TSAT\ formula $ \phi(x_1, \ldots , x_n) $ with clauses $ C_1, \ldots , C_m $ 
	we construct a simple drawing $D$ as described in Section~\ref{sec:hardness} and aim to insert the edge $uv$ into it. 
	This construction can clearly be computed in polynomial time and space, 
	since only the combinatorial description of the drawing is needed.  
	
	Assume $ uv $ can be inserted into $ D $ and let $ uv $ be the resulting arc.
	By Lemmas~\ref{lem:boundary} and~\ref{lem:blockers} we know that $ uv $ has to cross $ b_2^* $ and every arc in $ F $. 
	Let $ u^*$ be the point where $ uv $ crosses~$ b_2^* $.
	Each clause and variable gadget separates $ u^* $ from $ v $ and 
	thus, Lemmas~\ref{lem:variable} and~\ref{lem:clause} can be applied.
	This means that in a variable gadget $ W^{(i)} $ all arcs in $ P $ 
	or all arcs in $ N $ are crossed. 	
	In the former case we assign to variable $x_i$ the value \texttt{true}, 
	and otherwise the value \texttt{false}. 
	Assume that this truth assignment does not satisfy $ \phi(x_1, \ldots, x_n) $. 
	Then there exists a clause $ C_j $ for which all three literals evaluate to $ \texttt{false} $. 
	Consider the clause gadget~$K^{(j)}$. 
	By Lemma~\ref{lem:clause} we must cross in it an edge corresponding to one of its literals. 
	However, by Lemma~\ref{lem:blockers} an edge corresponding to the constant value $ \texttt{false} $ 
	cannot be crossed (again) in a clause gadget. 
	By construction and the truth assignment of the variables, 
	the edges corresponding to the other literals of $ C_j $ cannot be crossed either.

	Conversely, assume we are given a satisfying assignment of $ \phi(x_1, \ldots , x_n) $.
	We then can insert $ uv $ into $ D $ as follows.
	Starting from $u$, edge $uv$ crosses
	$ a_1 $ to enter region $ A_1 $, 
	then crosses all arcs in $ F $, 
	and crosses $b_2^*$ to enter $R$; 
	see also the dotted line in \figurename~\ref{fig:reduction}. 
	In each clause gadget, edge $ uv $ crosses 
	one edge corresponding to a literal evaluating to \texttt{true},  
	none corresponding to a literal evaluating to \texttt{false}, 
	and the edge $dg$ in the gadget if necessary. 
	By construction, this leaves in each variable gadget all arcs either in $ P $ or in $ N $ free to be crossed by $ uv $. 
	Moreover, this allows us to connect $ u$ and $ v $ without crossing any edge twice.
\end{proof}

As our reduction from \TSAT\ constructs a simple drawing $D(G)$ of a matching, 
the general problem is \NP-hard even if $G$ is as sparse as possible.
We remark that if we do not require $G$ to be a matching, 
our variable gadget can be simplified by identifying all the vertices on $\kappa$ 
and removing the crossings between edges in~$N$ and~$P$. 
Moreover, from the constructed drawing $D(G)$, one can produce an equivalent instance that is connected:
This is done by inserting an apex vertex into an arbitrary cell of the drawing, 
and then subdividing its incident edges so that the resulting drawing $D^*$ is simple. 
If $uv$ can be inserted into $D(G)$ 
then it can be inserted also into $D^*$\!.
Finally, in the next section we show that the problem remains hard even when the input drawing 
$D(G)$ is a pseudocircular drawing and we are in addition given an arrangement of pseudocircles extending $D(G)$,
regardless of whether the resulting drawing is required to be again pseudocircular or allowed be any simple drawing.

\section{Inserting one edge into a pseudocircular drawing is still hard}
\label{sec:extension}
In this section, we show that the simple drawings produced by our reduction are actually pseudocircular. Hence we obtain the following corollary.

\begin{restatable}{corollary}{corpseudocircle}\label{cor:pseudocircle}
	Given a pseudocircular drawing $ D(G) $ of a graph $ G = (V,E) $ and an edge $ uv $ of $\overline{G}$, it is \NP-complete to decide whether $ uv $ can be inserted into~$ D(G) $,
	even if an arrangement of pseudocircles extending the drawing of the edges in $ D(G) $ is provided.
\end{restatable}
\begin{proof}
	Let $ D $ be a drawing produced by our reduction from \TSAT.
	We divide the edges that correspond to literals of the input \TSAT-formula into 
	the \emph{blue edges} and the \emph{purple edges}.
	The former correspond to positive literals and 
	the latter to negative ones.
	Furthermore, we call the edges corresponding to constant \texttt{false} values the \emph{orange edges} and 
	the four edges $ r_1 $, $ r_2 $, $ \ell_1 $, and $ \ell_2 $ the \emph{red edges}.
	For each clause gadget we find one edge that is not corresponding to a literal or constant \texttt{false} value; we call all these edges the \emph{black edges}.
	Finally, we call the edges in the subdrawing $ \snail $ in $ D $ the \emph{green edges}. 
	
	To complete $ D $ into an arrangement of pseudocircles we have to close every 
	blue, purple, black, orange, red, and green edge by a corresponding \emph{extension}.
	For the six green edges this can be done as shown in \figurename~\ref{fig:pskyncl}.
	The orange and red edges are partitioned into two groups. 
	The first one contains $r_1$, $r_2$ and the orange edges corresponding to \texttt{false} values in clauses of Type~(i). 
	The second one  
	contains $\ell_1$, $\ell_2$ and the orange edges corresponding to \texttt{false} values in clauses of Type~(iv). 
	Inside the region $R$, 
	for both groups the red and the orange extensions are drawn as parallel, pairwise non-intersecting curves 
	between their endpoints in $ R $ and the boundary of the region $ A_2 $; 
	see \figurename~\ref{fig:psstructure}. 
	Additionally, also inside $ R $, for each group the extensions of the two red edges 
	cross all the orange edges in the group.  
	Moreover, the clause gadgets are essentially placed between the red extensions.
	Inside the region~$ A_2 $, 
	for each group the extensions of the two red edges cross and 
	the orange extensions cross the red ones; 
	see again \figurename~\ref{fig:psstructure}.

	\begin{figure}
		\centering
		\subfloat[The drawing $ \snail' $ extended to an arrangement of pseudocircles.]{\label{fig:pskyncl}\makebox[.48\textwidth][c]{\includegraphics[width=.35\textwidth,page=3]{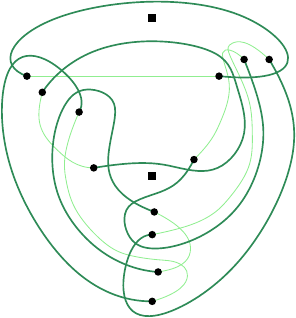}}}
		\hfill
		\subfloat[Extending the red and orange edges to pseudocircles.]{\label{fig:psstructure}\makebox[.48\textwidth][c]{\includegraphics[width=.45\textwidth,page=4]{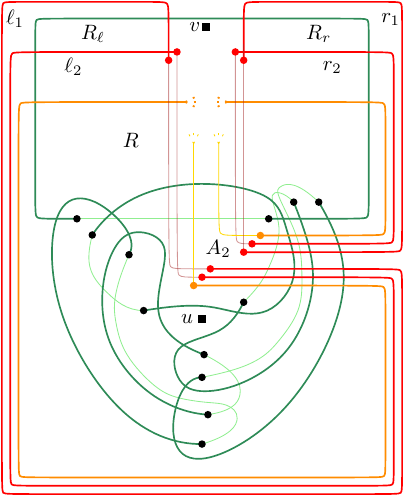}}}
		\caption{Extending the gadgets that form the frame of our reduction to an arrangement of pseudocircles.}
		\label{fig:psoverview}
	\end{figure}	
	
	We close the black edges with black extensions by just connecting the endpoints of a black edge
	without producing any additional crossings with the edges of $D$ or with the extensions defined so far.
	It remains to extend the purple and the blue edges.
	An example of a fully extended drawing $ D $ can be seen in \figurename~\ref{fig:pseudocircular}.
	The purple and blue extensions are essentially horizontally mirrored copies of their corresponding edges. 
	In particular, two purple or blue extensions cross if and only if the corresponding purple or blue edges cross.
	Moreover, inside the region $R$, the purple and the blue extensions are drawn without crossings. 
	As we traverse~$\ell_1$ from its endpoint in $A_2$ to its endpoint in $R$,  
	we encounter the (crossing points of) purple extensions of arcs in $ W^{(i)} $ 
	after the blue arcs in $ W^{(i-1)} $
	and before 
	the blue arcs in $ W^{(i)} $. 
	Analogously, as we traverse $r_1$ from its endpoint in $A_2$ to its endpoint in $R$,  
	we encounter the (crossing points of) blue extensions of arcs in $ W^{(i)} $ 
	after the purple arcs in~$ W^{(i-1)} $
	and before 
	the purple arcs in $ W^{(i)} $.
	Furthermore, as we traverse $\ell_2$ from its endpoint in $A_2$ to its endpoint in $R$,  
	we encounter the (crossing points of) purple extensions before the blue arcs. 
	Similarly, as we traverse $r_2$ from its endpoint in $A_2$ to its endpoint in $R$,  
	we encounter the (crossing points of) blue extensions before the purple arcs.

	Let $ D^\circ $ be the arrangement of closed curves constructed from $ D $.  
	It remains to prove that $ D^\circ $ is an arrangement of pseudocircles. 
	We consider the pseudocircles in $ D^\circ $ to have the same color as the edges and extensions that define them. 
	We first show that we can deform 
the purple, blue, black, red, and orange pseudocircles in $ D^\circ $ such that they are all axis-aligned rectangles and the pairwise intersections are preserved. 
Then, to show that two of these rectangles cross at most twice 
	we make use of the next observation:  
	
\setcounter{observation}{0}
	\begin{observation}\label{obs:rectangles_ap}
		Let $\square_1$ and $\square_2$ be two axis-aligned rectangles whose vertices lie in general position (no three are collinear).  
		If the leftmost and rightmost points of the projection of $\square_1 \cup \square_2$ 
		into the horizontal (or vertical) axis correspond to different rectangles, 
		then $\square_1$ and $\square_2$ cross in at most two points. 
	\end{observation}

	\begin{figure}
		\centering
		\includegraphics[page=2]{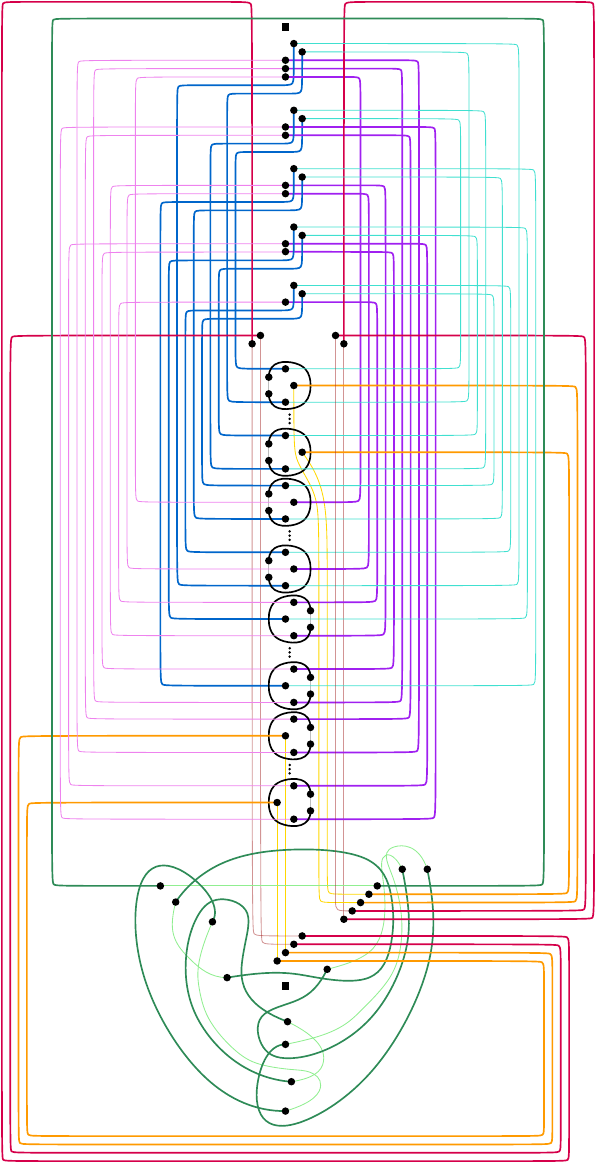}
		\caption{The drawing produced by our reduction is pseudocircular.}
		\label{fig:pseudocircular}
	\end{figure}

We will show that all pseudocircles in $ D^\circ $ except the green ones can be deformed to axis-aligned rectangles while maintaining their intersections with other pseudocircles. 
We refer to Figure~\ref{fig:pseudocircular}. 	

By construction,  the red and orange pseudocircles 
extending the edges in the group of red and orange ones that
contains $ r_1 $ and $r_2$ can be drawn directly as axis-aligned rectangles. 
See the the red and orange pseudocircles on the right side of Figure~\ref{fig:pseudocircular}. 
We deform (the bottom part of) the other orange and red pseudocircles 
such that the resulting pseudocircles are axis-aligned rectangles.
This can easily be done by also deforming 
part of the subdrawing $ \snail'$ of~$ D^\circ $. 

The purple pseudocircles can be drawn directly as axis-aligned rectangles. 
A black pseudocircle $\Phi$ extending a black edge $e$ can trivially be drawn as an axis-aligned rectangle such that $\Phi$ only crosses pseudocircles extending edges that cross $e$. 

We now deform the blue pseudocircles. 
The blue extensions as described above can be drawn such that the resulting 
blue pseudocircles are axis-aligned polygons with one reflex corner 
(between $\ell_1$ and $r_1$). 
For a blue pseudocircle~$\Phi$ drawn in this way, 
let the \emph{corner point} be the reflex vertex of the polygon and
let the horizontal and vertical sides incident with it be the 
\emph{horizontal corner-arc} and the \emph{vertical corner-arc} of $\Phi$, 
respectively. 
To make a blue pseudocircle an axis-aligned rectangle, 
we deform it by
moving its corner point;  
see \figurename~\ref{fig:bpcases}.
Obviously, this does not change the crossings with any green, black, red, or orange pseudocircle.
Furthermore, it does not change the crossings with other blue pseudocircles 
as no new crossings are introduced and 
the crossings along the horizontal corner-arc are preserved. 
Finally, in the same way, this deformation preserves the crossings 
between the blue pseudocircle and purple ones along the vertical corner-arc.

\begin{figure}[tb]
	\centering
	\subfloat[Blue pseudocircles.]{\label{fig:bluecase}\makebox[.3\textwidth][c]{\includegraphics[page=5]{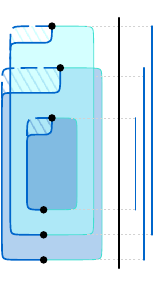}}
	}\hfill \subfloat[Blue and purple pseudocircles.]{\label{fig:bpnotsamevar}\makebox[.6\textwidth][c]{\includegraphics[page=6]{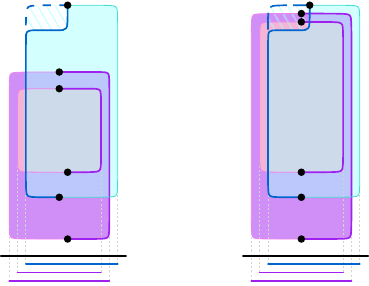}}
	}\caption{Interactions between the blue and the purple pseudocircles.}
	\label{fig:bpcases}
\end{figure}

Consider the deformed drawing obtained from $ D^\circ $ maintaining all intersections. 
We now argue that each two pseudocircles cross either zero or two times in this deformed drawing and hence in $ D^\circ $. 
To show that no two blue (or no two purple) rectangles cross more than twice 
we consider their projection onto the vertical axis. 
Then, by construction, two rectangles cross if and only if the topmost and the bottommost points of the projection correspond to different rectangles; see \figurename~\ref{fig:bluecase}. 
Thus, by Observation~\ref{obs:rectangles_ap}, 
in case the two rectangles cross they cross twice. 
For a blue and a purple pseudocircle we find that their projection to the horizontal axis
is always such that the left-most point belongs to the purple extension and 
the right-most point to the blue extension by construction; 
see \figurename~\ref{fig:bpnotsamevar} for an illustration.
From Observation~\ref{obs:rectangles_ap} it follows that each pair of blue and purple rectangles crosses at most twice.

In the same manner we can argue about the red and orange rectangles. 
By construction, two orange rectangles do not cross. 
A red and an orange rectangle are either disjoint 
(if they extend edges in different groups of red and orange ones) 
or the leftmost and rightmost points of their projection onto the horizontal axis 
correspond to different rectangles. 
Thus, from Observation~\ref{obs:rectangles_ap} it follows that each pair of red and orange rectangles crosses at most twice. 
Similarly, 
given a red or orange rectangle and a purple or blue one, 
the leftmost and rightmost points of their projection onto the horizontal axis 
correspond to different rectangles. 
Thus, by Observation~\ref{obs:rectangles_ap}, they cross at most twice.

Given two rectangles, one of them black, 
their projection onto the horizontal or the vertical axis shows that 
either they do not cross or, 
by Observation~\ref{obs:rectangles_ap}, they cross at most twice. 
Finally, it is easy to verify that no red, orange or green pseudocircle 
crosses a green pseudocircle more than twice.
Since by construction no other pseudocircle crosses a green pseudocircle, 
we conclude that $ D^\circ $ is in fact an arrangement of pseudocircles.
\end{proof}

\section{Extending an arrangement of pseudocircles is easy}
\label{sec:arr_pseudocicles}
\begin{figure}[b]
	\begin{minipage}[c]{.48\textwidth}
		\centering
		\includegraphics[page=1]{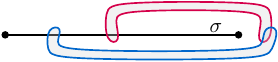}	
		\caption{Obstruction where all pseudocircles intersect $ \sigma $ twice.}
		\label{fig:obstruction_ext_segment_a2}
	\end{minipage}
	\hfill
	\begin{minipage}[c]{.48\textwidth}
		\centering
		\includegraphics[page=2]{obstruction_extending_segment}
		\caption{Obstruction where one pseudocircle intersects $ \sigma $ only once.}
		\label{fig:obstruction_ext_segment_a1}
	\end{minipage}
\end{figure}
In Section~\ref{sec:hardness} we proved that deciding whether an edge can be inserted into a pseudocircular drawing such that the result is a simple (or a pseudocircular) drawing is hard.
In this section we focus on extending arrangements of pseudocircles instead of drawings of graphs. 
Recall that in such an arrangement the restriction is that two pseudocircles can cross at most twice while in a simple drawing the restriction is that two edges share at most one point. 
The main difference in extending arrangements of pseudocircles and simple pseudocircular drawings is that in the latter the crossing possibilities are more restricted: the \emph{edge parts} of two pseudocircles cannot cross twice.

Snoeyink and Hershberger~\cite{sweeping_pseudocircles91} showed that given an arrangement $\mathcal{A}$ of pseudocircles and three points, not all three on the same pseudocircle, 
one can find a pseudocircle $\Phi$ through the three points such that $\mathcal{A}\cup \{\Phi\}$ is again an arrangement of pseudocircles. 
Now, given any arrangement~$\mathcal{A}$ and a pseudosegment $\sigma$ intersecting each pseudocircle in ${\cal A}$ at most twice, 
it is not always possible to extend  $\sigma$ to a pseudocircle $\Phi_\sigma \supset \sigma$ such that  ${\cal A} \cup \{\Phi_\sigma\}$ is again an arrangement of pseudocircles.
Two examples are shown in Figures~\ref{fig:obstruction_ext_segment_a2} and~\ref{fig:obstruction_ext_segment_a1}. 
In both examples any pseudocircle $\Phi_\sigma$ extending $\sigma$  crosses 
one red or blue pseudocircle at least four times.
We show in the following that the extension decision question can be answered in polynomial time:

\begin{theorem}	\label{thm:pseudoeasy}
	Given an arrangement $\mathcal{A}$ of $n$ pseudocircles and a pseudosegment~$\sigma$ intersecting each pseudocircle in ${\cal A}$ at most twice,
	it can be decided in time polynomial in $n$ whether there exists an extension of $\sigma$ to a pseudocircle $\Phi_\sigma$ such that 
	that ${\cal A} \cup \{\Phi_\sigma\}$ is an arrangement of pseudocircles. 
\end{theorem}

An arrangement (of pseudocircles) partitions the plane into \emph{vertices} (0-dimensional cells), 
\emph{edges} (1-dimensional cells), and \emph{faces} (2-dimensional cells). 
Since tangencies are not allowed, all vertices are proper crossings.
Note that an arrangement of $n$ pseudocircles has $O(n^2)$ complexity. 
Two arrangements are \emph{combinatorially equivalent} (or, \emph{isomorphic}) 
if the corresponding cell complexes are isomorphic, that is, 
if there is an incidence- and dimension-preserving bijection between their cells. 
The extention problem does not depend on the particular geometry of the arrangement, only on the combinatorial equivalence class. 
Therefore, we can assume that the input is this combinatorial description (of polynomial size in $n$).

\begin{proof}[of Theorem~\ref{thm:pseudoeasy}] 
	Throughout this proof we write $ \overline{R} := \mathbb{R}^2 \setminus R $ 
	for the \emph{complement} of a set $ R \subseteq \mathbb{R}^2 $.
By possibly transforming $\mathcal{A}$ into an isomorphic arrangement
	while preserving the incidences of $\sigma$, 
	we can assume without loss of generality that an endpoint is incident with the unbounded cell and
	that the intersection points of~$\sigma$ with the pseudocircles in $\mathcal{A}$ 
	are all proper crossings.
	Further, by possibly transforming the arrangement again into an isomorphic one, we can assume that $\sigma$ is a horizontal segment with the left endpoint incident with the unbounded cell.  
	Let $u$ and $v$ be the left and right endpoints of $\sigma$, respectively.
	Our algorithm aims to compute a pseudocircle $\Phi_\sigma = \sigma \cup \sigma'$ such that ${\cal A} \cup \{\Phi_\sigma\}$ is an arrangement of pseudocircles, or determine that no such $\sigma'$ exists.
We call $\sigma'$ an \emph{extension} of~$\sigma$.
	
	We partition the set of pseudocircles of $\mathcal{A}$ into three sets $\mathcal{C}_0 $, $\mathcal{C}_1$, and $\mathcal{C}_2$, 
	where for each $i\in\{0,1,2\}$, $\mathcal{C}_i$ is the set of pseudocircles in $\mathcal{A}$ crossing $\sigma$ exactly~$i$ times.
Note that $u$ lies outside all pseudocircles $\phi \in \mathcal{A}$ 
	while $v$ lies outside of all $\phi \in \mathcal{C}_0 \cup \mathcal{C}_2$ and inside all $\phi \in \mathcal{C}_1$,
	that is, each $\phi \in \mathcal{C}_1$ separates $u$ and $v$.
	Further, an extension $\sigma'$ must not cross any $\phi \in \mathcal{C}_2$, it needs to cross every $\phi \in \mathcal{C}_1$ exactly once,
	and it can cross each $\phi \in \mathcal{C}_0$ either twice or not at all.
	
	\begin{figure}
		\begin{minipage}[c]{.48\textwidth}
			\includegraphics[page = 1]{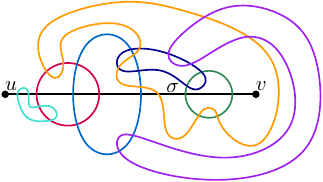}
			\subcaption{Initial arrangement of pseudocircles $\mathcal{A}$ and pseudosegment $\sigma$.}
			\label{subfig:1_ext}			
		\end{minipage}
		\begin{minipage}[c]{.48\textwidth}
			\includegraphics[page = 2]{extending_segment}
			 \subcaption{Simply-connected subset $R_0$.}
			 \label{subfig:2_ext}
		\end{minipage}
	
		\begin{minipage}[c]{.48\textwidth}
			\includegraphics[page = 3]{extending_segment}
			\subcaption{Simply-connected subset $R_1$.}
			\label{subfig:3_ext}
		\end{minipage}
		\begin{minipage}[c]{.48\textwidth}
			\includegraphics[page = 4]{extending_segment}
	 		\subcaption{Simply-connected subset $R_m$ and two possible extensions $\sigma'_1$ and $\sigma'_2$.}
	 		\label{subfig:4_ext}
		\end{minipage}		
		\caption{Algorithm extending $\sigma$ to a pseudocircle $\Phi_\sigma$.}
		\label{fig:algorithm_extending}
	\end{figure} 
	
The idea is to construct a finite sequence $R_0\subset R_1\subset \ldots$ of closed subsets of~$\mathbb{R}^2$, each consisting of cells of $\mathcal{A} \cup \sigma$ that cannot be reached by $\sigma'$. \figurename~\ref{fig:algorithm_extending} illustrates this idea as well as various cases throughout the proof.
	Each set~$R_i$ will be a simply connected closed region of $\mathbb{R}^2$ with both $u$ and $v$ on its boundary. 
Further, we will maintain the following \emph{invariant}: 

\begin{description}
	\item[]\!\!\emph{For each $R_i$ and each $\phi \in \mathcal{C}_0$,
	$\int(\phi) \cap \overline{R_i}$ is either a connected region or empty,}
\end{description}
 where  $\int(\phi)$ denotes the interior of the bounded area enclosed by $\phi$.
The construction will either end by determining that $\sigma$ cannot be extended, or with a set $R_m$ such that routing $\sigma'$ closely along the boundary of $R_m$ gives a valid extension of $\sigma$.

Let  $R'_0$ be the union of $\sigma$ and all the closed disks bounded by the pseudocircles in $\mathcal{C}_2$ and consider the faces induced by $ R'_0 $. 
	Since $u$ is incident with the unbounded cell of $R'_0 $, and since $\sigma'$ must not intersect the interior of~$R'_0$,
	$\sigma'$ cannot reach any bounded face of  $R'_0 $. 
Let $R_0$ be the closure of the union of these bounded faces and $\sigma$. 
	We may assume that $v\in \partial R_0$, as otherwise no extension~$\sigma'$ exists and we are done.

	To see that the {invariant} holds for $ R_0 $, assume that there exists a pseudocircle $\phi \in \mathcal{C}_0$ such that $\int(\phi) \cap \overline{R_0}$ is connected; 
	see \figurename~\ref{fig:invariant_R0} for an illustration. 
	Note that $\int(\phi) \cap \overline{R_0}$ is connected if and only if $R_0 \setminus \int(\phi)$ is connected.
	As $\phi$ does not intersect $\sigma$, there exists a component $D$ of $R_0 \setminus \int(\phi)$ that is disjoint from $\sigma$. 
	Further, as $\int(\phi)$ is simply connected, $D \cap \partial R_0 \neq \emptyset$.
	Moreover, any point $x$ on $\partial D \cap \partial R_0$ lies on some pseudocircle $\phi_x \in \mathcal{C}_2$.
	On the other hand, any path in $R_0$ from a point of $\sigma$ to $x$ must enter and leave $\int(\phi)$ and hence intersect $\phi$ at least twice. As $\phi_x$ intersects $\sigma$ twice and lies in $R_0$, 
we get that $\phi_x$ intersects $\phi$ in at least four points, a contradiction.
\begin{figure}
		\centering
		\includegraphics{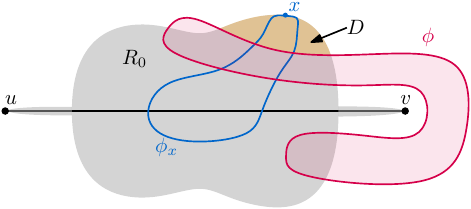}
		\caption{Proving that $R_0$ fulfills the invariant.}
		\label{fig:invariant_R0}
	\end{figure}

	For the iterative step, consider the arrangement $\mathcal{A}^{\phi}_i$ formed by $\partial R_i$ and a pseudocircle $\phi \in \mathcal{C}_0 \cup \mathcal{C}_1$,
	and the cells of it that lie in $\overline{R_i}$.
	If $\phi \in \mathcal{C}_1$ and an extension $\sigma'$ exists, then the only two such cells that can be intersected by~$\sigma'$ are the ones incident to $u$ and $v$, respectively.
	Similarly, if $\phi \in \mathcal{C}_0$, then $\sigma'$ can only intersect the cell(s) incident to $u$ and $v$, 
	plus the (by the {invariant}) unique cell $\int(\phi) \cap \overline{R_i}$.
	In both cases, all other cells of this arrangement should be added to the forbidden area. 
	We denote all cells  $\mathcal{A}^{\phi}_i \cap \overline{R_i}$ that can possibly be intersected by $\sigma'$ as \emph{reachable} (by $\sigma'$) and all other cells as \emph{unreachable} (by~$\sigma'$).

Assume that there exists some pseudocircle $\phi \in \mathcal{C}_0 \cup \mathcal{C}_1$ 
	such that the arrangement~$\mathcal{A}^{\phi}_i$ of $\phi$ and $\partial R_i$ contains unreachable cells.
Then we obtain $R'_{i+1}$ by adding all those cells to $R_i$. 
	If $v$ lies in a bounded region of $\overline{R'_{i+1}}$, then no extension~$\sigma'$ exists and we are done. 
	(Recall that by assumption $u$ always lies in the unbounded region.)
	Otherwise, $R_{i+1} = R'_{i+1}$ is a simply connected region that has both $u$ and $v$ on its boundary.
It remains to show that the {invariant} is still maintained for $R_{i+1}$. 

\begin{restatable}{lemma}{lemextendoneedge}\label{lem:extendoneedge}
	If $R_i$ fulfills the invariant and $u$ and $v$ both lie in the unbounded region of $\overline{R'_{i+1}}$ then $R_{i+1}$ also fulfills the invariant.
\end{restatable}
\begin{proof}
Let $\phi \in \mathcal{C}_0 \cup \mathcal{C}_1$ be the pseudocircle that causes the step from $R_i$ to $R_{i+1}$
	and consider the arrangement~$\mathcal{A}^{\phi}_i$ of $\phi$ and $\partial R_i$ (which contains unreachable cells).
Note that the boundaries of all cells of $\mathcal{A}^{\phi}_i$ alternate between arcs of $\phi$ and parts of $\partial R_i$.
	Moreover, all cells of $\mathcal{A}^{\phi}_i$ in $R_{i+1} \setminus R_i$ are bounded.

	We first consider the case that $\phi \in \mathcal{C}_0$. 
	It is illustrated in \figurename~\ref{fig:invariantC0}.
Suppose that there exits a pseudocircle $\phi'\in\mathcal{C}_0$ for which $\int(\phi') \cap \overline{R_{i+1}}$ is disconnected while $\int(\phi') \cap \overline{R}_{i}$ is connected. 
	Observe that $\phi'\neq \phi$ because all the cells of $\mathcal{A}^{\phi}_i$ that are added to $R_i$ for obtaining $R_{i+1}$ lie outside $\phi$. 
	Since $R_i$ fulfills the invariant, each cell of $\mathcal{A}^{\phi}_i$ in $R_{i+1} \setminus R_i$ is bounded by a single arc of $\phi$ and a single arc of $\partial R_i$ and all those cells are pairwise disjoint. 
	Hence there exists at least one such cell $c$ that disconnects $\int(\phi') \cap \overline{R}_{i}$, and the boundary of $c$ along $\phi$ intersects $\phi'$ (at least) twice.
	Recall that $c$ is bounded and to the exterior of~$\phi$. 
	If $\phi'$ was only intersecting $\phi$ at those two points, 
	the boundary of $\phi'$ outside $c$ would be completely contained in $\int (\phi)$, but then $c$ would not disconnect $\int(\phi') \cap \overline{R}_{i}$. 
Thus, 
	$\phi$ must intersect $\phi'$ in at least two more points, a contradiction.
	
	Now consider the case $\phi \in \mathcal{C}_1$.
	For an illustration consider \figurename~\ref{fig:invariantC1}.
	Assume again that there exists a pseudocircle $\phi' \in \mathcal{C}_0$ for which $\int(\phi') \cap \overline{R_{i+1}}$ is disconnected while $\int(\phi') \cap \overline{R}_{i}$ is connected.
	Consider again a cell $c$ of $\mathcal{A}^{\phi}_i$ that is part of $R_{i+1} \setminus R_i$ and disconnects $\int(\phi') \cap \overline{R}_{i}$.
	The cell $c$ must not contain any of $u$ and $v$ as otherwise it would not be in $R_{i+1}$.
	Further, the cell $c$ cannot separate $u$ and $v$, as otherwise $v$ would have been in a bounded region of $R_{i+1}'$ and we would have stopped the process. 
As $c$ disconnects $\int(\phi') \cap \overline{R}_{i}$, 
	$\phi$ intersects $\phi'$ twice along the boundary of $c$ (and hence outside $R_i$).
As every pair of pseudocircles have at most two intersection points, 
	$\phi$ does not intersect $\phi'$ in any other points.
	Especially, $\phi$ does not intersect $\phi'$ inside $R_i$. 
	Furthermore, 
	$\phi$ intersects $\partial R_i$ in $\int(\phi')$ at least twice along $\partial c$ (causing the disconnection of  $\int(\phi') \cap \overline{R}_{i}$) and $\phi$ also intersects $\partial R_i$ outside of $\phi'$ (as it must intersect $\sigma$ and $\phi'$ cannot intersect $\sigma$). 
	This last property implies that each 
	component of $\int(\phi') \cap \overline{R_{i+1}}$ induced by $c$ lies in a different reachable cell of $\mathcal{A}^{\phi}_i$ that is neighboring to $c$ via an arc of $\phi$. 
	However, as $c$ does not separate $u$ and $v$, at most one such cell can exist, a contradiction to $\int(\phi') \cap \overline{R_{i+1}}$ being disconnected.
\end{proof}
\begin{figure}
		\subfloat[]{\label{fig:invariantC0}\makebox[.45\textwidth][c]{\hspace{-1.0cm}\includegraphics[page=1]{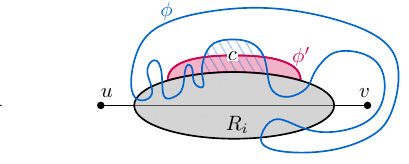}}}\hfill
\subfloat[]{\label{fig:invariantC1}\makebox[.45\textwidth][c]{\includegraphics[page=2]{invariant2}}}\caption{Illustration of potentially separating cells in the proof of Lemma~\ref{lem:extendoneedge}. 
			The red arc and area belong to $ \phi' \in \mathcal C_0 $, the blue striped area is the cell $ c $, the blue curve is the pseudocircle $ \phi \in \mathcal C_0 \cup \mathcal C_1$.}
		\label{fig:invariant}
	\end{figure}

	Now assume that both $u$ and $v$ lie on the boundary of all sets $R_i$ constructed in this way.
	Then the iterative process stops with a set $R_m$ where for each $\phi \in \mathcal{C}_0 \cup \mathcal{C}_1$, 
all cells in the arrangement $\mathcal{A}^{\phi}_m$ of~$\phi$ and $\partial R_m$ 
	that are contained in~$\overline{R_m}$ are reachable by $\sigma'$.
	Note that $m = O(n^2)$ as $\mathcal A$ has $O(n^2)$ cells, 
	in every iteration $i$ at least one cell of~$\mathcal A$ has been added to $R_i$, and each cell of~$\mathcal{A}$ is added at most once.
	Consider a path $P$ from $u$ to $v$ in $\overline{R_m}$ that is routed closely along the boundary $\partial R_m$ (note that there are two different such paths).
Then for any $\phi \in \mathcal{C}_1$, $P$ intersects exactly two cells of $\mathcal{A}^{\phi}_m$, 
	namely, the ones incident to $u$ and~$v$, respectively. 
	Hence $P$ crosses $\phi$ exactly once.
	Similarly, for any $\phi \in \mathcal{C}_0$, the path $P$ intersects at most three cells of $\mathcal{A}^{\phi}_m$, 
	namely, the one(s) incident to $u$ and~$v$ plus possibly the cell $\int(\phi) \cap \overline{R_m}$, 
	which is one cell by the {invariant}. Hence $P$ crosses~$\phi$ at most twice.
Thus $\sigma'=P$ is a valid extension for~$\sigma$, which completes the correctness argument.
	
Note that computing $R_0$ and $\sigma'$ (in case that the algorithm didn't terminate with a negative answer before) can be done in polynomial time.
	Also, for each~$ R_i $ and each $ \phi \in \mathcal C_0 \cup \mathcal C_1 $, the set of unreachable cells
	of $ \mathcal{A}^{\phi}_i $ can be determined in polynomial time.
	As we have $ O(n^2) $ iteration steps, we can hence compute $ R_{m} $ from $ R_0 $ (or determine that $ \sigma $ is not extendible) in polynomial time, which concludes the proof.
\end{proof}

As an immediate consequence of Theorem~\ref{thm:pseudoeasy} we have the following result:

\begin{corollary}
Given an arrangement $\mathcal{A}$ pseudocircles and a pseudosegment~$\sigma$, 
it can be decided in polynomial time whether $\sigma$ can be extended to a pseudocircle $\Phi_\sigma \supset \sigma$ such that ${\cal A} \cup \{\Phi_\sigma\}$ is an arrangement of pseudocircles.
\end{corollary}

\section{An FPT-algorithm for bounded number of crossings}
\label{sec:fpt}
In this section we show that for drawings with a bounded number of crossings 
it can be decided in \FPT-time whether an edge can be inserted. 
Given a simple drawing $ D(G) $ with $k$ crossings,
one can construct a \emph{kernel} of size $ O(k) $ by exhaustively removing isolated vertices and uncrossed edges from $ D(G) $.
For a simple drawing $D(G)$ of a graph $ G = (V,E) $ and $ e \in E $, let
$ D(G - e) $ be the subdrawing of $ D(G) $ without the drawing of $ e $.
Similarly, for an isolated vertex $ u \in V $ let $ D(G - u) $ be the subdrawing of $ D(G) $ 
without the drawing of $ u $.

\begin{observation}\label{obs:isolated}
	Given a simple drawing $ D(G) $ of a graph $ G = (V,E) $ and an isolated vertex $w \in V$, an edge $ uv $ of $\overline{G}$ can be inserted into~$ D(G) $ if and only if $ uv $ can be inserted into $ D(G - w) $.
\end{observation}

By Observation~\ref{obs:isolated} we get that isolated vertices can be disregarded in an algorithm
that extends a simple drawing $D(G)$ of a graph by one edge. The following lemma implies that the same is true for uncrossed edges in $ D(G)$.

\begin{restatable}{lemma}{lemuncrossed}\label{lem:uncrossed}
	Given a simple drawing $ D(G) $ of a graph $ G = (V,E) $ and an edge $ e \in E $ that is uncrossed in $ D(G) $,
	an edge $ uv $ of $\overline{G}$ can be inserted into~$ D(G) $ if and only if $ uv $ can be inserted into $ D(G - e) $.
\end{restatable}
\begin{proof}
Since $D(G - e)$ is a subdrawing of $D(G)$, 
	it is clear that 
if $uv$ can be inserted into~$D(G)$ 
	then it can be inserted into $D(G - e)$. 
	Suppose that $uv$ can be inserted into $D(G - e)$ and 
	let~$\gamma$ be a \emph{valid drawing} of $uv$ in $D(G - e)$, that is, one resulting in a simple drawing of $G \setminus \{e\} \cup \{uv\}$.
	We orient $\gamma$ from~$u$ to~$v$.  
	If~$\gamma$ is not a valid drawing of $uv$ in $D(G)$ 
	then it must intersect $ e$ more than once in $ D(G) $.
We can modify $\gamma$ such that it is routed close to $e$ between its first and last intersection with $e$,  
	producing at most one intersection; 
	see \figurename~\ref{fig:reroute} for an illustration. 
	If $ e $ is not incident to $ u $ or $ v $ we are done.
	Else assume without loss of generality that $ e $ is incident to $ u $
	and let $ \gamma' $ be the drawing of~$ uv $ that was modified such that it has only one intersection with~$ e $.
Recall that $ e $ is uncrossed in $ D(G) $.
	Hence, the intersection point $\times$ of $\gamma'$ with~$e$ and the point~$ u $ lie on the boundary of one cell in $ D(G) $.
	Consequently, we can modify $ \gamma' $ in such a way that it is routed closely to $ e $ from $u$ to $\times$ on the side of~$e$ on which~$\gamma'$ continues to $v$ 
	without producing a crossing with any other edge in $ D(G) $.
	Either modification only reduces crossings, but does not introduce new ones,
	hence we obtained a valid drawing of $ uv $ in $ D(G) $ as desired.
\end{proof}

\begin{figure}
	\centering
	\includegraphics[page=1]{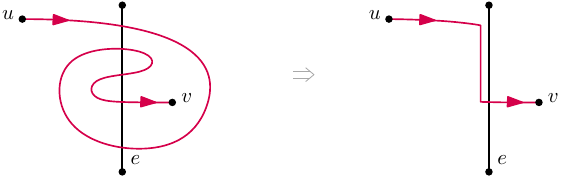}
	\caption{Rerouting $uv$ when it crosses an otherwise uncrossed edge more than once.}
	\label{fig:reroute}
\end{figure}

Equipped with Observation~\ref{obs:isolated} and Lemma~\ref{lem:uncrossed}
we are ready to prove the main theorem of this section.

\begin{restatable}{theorem}{thmfptalg}\label{thm:fptalg}
	Given a simple drawing $ D(G) $ of a graph $ G = (V,E) $ and an edge $ uv $ of $\overline{G}$, there is an \FPT-algorithm in the number $k$ of crossings in $D(G)$ 
	for deciding whether $ uv $ can be inserted into $ D(G) $. 
\end{restatable}
\begin{proof}
Let $ G' $ be the subgraph of $ G $ remaining after exhaustively deleting uncrossed edges and isolated vertices distinct from $ u $ and $ v $.
	Furthermore, let $ D'(G') $ be the corresponding subdrawing of $ D(G) $.
	By assumption, there are at most $ 2k $ crossed edges in $ G $. Hence $ G' $ has at most $ 4k + 2 $ vertices and $ 2k $ edges.
	Furthermore, by Observation~\ref{obs:isolated} and Lemma~\ref{lem:uncrossed} we can insert $ uv $ into $ D(G) $
	if and only if it can be inserted into $ D'(G') $.
	
	For solving the kernel instance of inserting $uv$ into $D'(G')$,
	we reformulate the problem of inserting an edge 
	into a simple drawing as a problem in the dual graph of its planarization,
	as in~\cite{arroyo_gd_2019}.
	In the planarization crossings are replaced by vertices resulting in a plane drawing. 
	Given a simple drawing $D(G)$ of a graph~$G$, 
	the {\em dual graph} $G^*(D)$ is the plane dual of the planarization of $D(G)$. 
	Thus, every vertex in $G^*(D)$ corresponds to a cell in $D(G)$ and 
	every edge in $G^*(D)$ corresponds to a segment of an edge in $D(G)$.  
	We assign to each edge in $D(G)$ a different color (label)
	and define a coloring $\chi$ of the edges of $G^*(D)$,
	where every edge in $G^*(D)$ inherits the color of its primal edge in $D(G)$.
	Given two vertices $u,v\in V$, 
	let $G^*(D,\{u,v\})$ be the subgraph of $G^*(D)$ obtained by removing from it the edges 
	corresponding to segments of edges incident with $u$ or to $v$. 
	Let $\chi'$ denote the coloring of the edges of $G^*(D,\{u,v\})$ that coincides with $\chi$ in every edge. 
	The problem of extending $D(G)$ with one edge $uv$ is then equivalent to
	the problem of finding
a path in $G^{*}(D,\{u,v\})$ 
	between a vertex corresponding to a cell incident with $u$ and a vertex corresponding to a cell incident with $v$ in which no color given by $\chi$ is repeated (that is, 
	the path is \emph{heterochromatic}). 
	
	The number of segments of crossed edges in $D'(G')$ is at most $4k$. 
	Thus, $G^*(D',\{u,v\})$ has at most $4k$ edges 
	(while the number of vertices might not be bounded by a function of $k$). 
	There are $O(n)$ cells in $D'(G')$ with $u$ or $v$ on their boundary. Further, every cell in  $D'(G')$ has complexity $O(k)$.
	Checking whether $uv$ can be inserted into $D'(G')$ can be done by 
	(i) checking for each of the $O(n)$ vertices in $G^*(D',\{u,v\})$ whether both $u$ and $v$ are incident to the according cell in $D'(G')$ 
	and (ii) checking for each of the $O(2^{4k})$ non-empty subsets of edges in $G^*(D',\{u,v\})$ whether they form a valid heterochromatic path with endpoints incident to $u$ and $v$, respectively. Altogether, this can be done (brute-force) in $O(nk + k^2 2^{4k})$ time. 
\end{proof}

\section{Conclusions}
In this paper we showed that given a simple drawing $D(G)$ of a graph $G$ 
it is \NP-hard to decide if a particular edge from the complement of $G$ 
can be inserted into $D(G)$ such that the result is a simple drawing.
On the positive side, we showed that for a given pseudocircular arrangement $\mathcal A$ of pseudocircles and 
a pseudosegment $\sigma$ it can be decided in polynomial time whether $\sigma$ can be extended to a simple
closed curve $\Phi_\sigma$ such that $\mathcal A \cup \{\Phi_\sigma\}$ is again an arrangement of pseudocircles.
Furthermore, we proved that the problem is \FPT\ with respect to the number of crossings of $D(G)$.

In the light of our results,
checking whether a simple drawing $D(G)$ is saturated 
by trying to insert every edge of the complement of $G$ 
is hopeless (unless $\P = \NP$). 
Thus, it is an interesting open problem 
whether there is a polynomial algorithm for deciding 
if a simple drawing is saturated.

\bibliographystyle{plainurl}
\bibliography{add_one_edge}

\end{document}